\documentclass[11pt]{article}

\usepackage[a4paper,margin=27mm]{geometry}
\usepackage[utf8]{inputenc}
\usepackage[T1]{fontenc}
\usepackage{amsthm,amsmath,amsfonts,amssymb}
\usepackage{mathtools}
\usepackage{graphicx}
\usepackage{booktabs}
\usepackage{geometry}
\usepackage{authblk}
\usepackage{natbib}
\usepackage{hyperref}
\usepackage{caption}
\usepackage{float}
\usepackage{setspace}
\usepackage{microtype}
\usepackage{xcolor}
\usepackage{longtable}

\hypersetup{colorlinks=true,citecolor=blue,linkcolor=blue,urlcolor=blue}
\microtypesetup{expansion=false}
\let\footnote\thanks

\providecommand{\eqref}[1]{(\ref{#1})}
\providecommand{\figrule}{}

\theoremstyle{definition}
\newtheorem{rem}{Remark}[section]
\theoremstyle{plain}
\newtheorem{defn}{Definition}[section]
\newtheorem{prop}{Proposition}[section]
\newtheorem{thm}{Theorem}[section]
\newtheorem{cor}{Corollary}[section]
\newtheorem{lem}{Lemma}[section]

\newcommand{\Alts}{\mathcal{X}}
\newcommand{\V}{V}
\newcommand{\E}{E}
\newcommand{\R}{\mathbb{R}}
\newcommand{\one}{\mathbf{1}}
\newcommand{\Reach}{\mathrm{Reach}}

\title{Borda Aggregation Dynamics of Preference Orderings on Networks}
\author{Moses Boudourides \\[0.5em]
{\small School of Professional Studies, Northwestern University, Evanston, IL, USA} \\
{\small \texttt{Moses.Boudourides@northwestern.edu}}}
\date{}

\begin{document}
\maketitle

\begin{abstract}
We introduce and analyze a discrete-time network process in which each node holds a (weak) preference ordering over a finite set of alternatives and updates by local \emph{Borda aggregation}. At each step, a node forms a weighted average (row-stochastic random-walk normalization) of its neighbors' Borda score vectors and projects the aggregated score back to a weak order. 
Updates are \emph{bounded}: in each round, a node advances by at most one step
along a shortest path in the preference graph toward its local Borda target,
reflecting incremental preference revision.
Our emphasis is dynamical: we develop sufficient conditions, stated directly in terms of graph topology and weights, for (i) \emph{self-sustained oscillations} in the absence of persistent sources, and (ii) \emph{forced oscillations} under contrarian persistent camps. We show that these oscillatory phenomena are robust features of the Borda-aggregation coupling, driven by the topology of the influence network and the geometry of the preference space rather than by the bounded-step rule per se. We also record robustness (structural stability) away from score-tie hyperplanes and contrast synchronous (Variant~S) and asynchronous (Variant~A) updating.
\end{abstract}


\section{Introduction and social-scientific motivation}
\label{sec:intro}

A wide class of diffusion and opinion-formation models treats each node's state as a scalar or vector and updates by averaging or bounded-confidence rules
\citep{degroot1974,friedkin1990,hegselmann2002,acemoglu2011}.
Across social and political science, however, many outcomes are intrinsically \emph{ordinal}:
actors rank candidates, policies, coalition partners, or consumption bundles, often allowing ties.
In such settings the primitive state is a preference ordering rather than a point in Euclidean space.

While classical opinion-dynamics models typically represent individual
states as scalar or vector quantities, such cardinal representations
may fail to capture the intrinsically ordinal nature of many
decision-making processes. In settings where agents rank alternatives,
possibly with indifferences, reducing preferences to numerical values
can obscure the combinatorial structure of rankings and the discrete
transitions between them. The present framework instead models
preferences directly as orderings and studies their evolution under
network interaction.

The social space sustaining these processes is naturally \emph{relational}.
Actors interact through a social influence network whose edges encode exposure, authority, credibility, or control;
formally, this empirical network is typically directed and weighted, and analytically it is often modeled through a row-stochastic influence matrix
\citep{french1956,harary1959,abelson1964,degroot1974,friedkin1990,proskurnikovtempo2017}.
From this viewpoint, the state of the system at time $t$ is a distributed configuration of individual outcomes over the network nodes,
dynamically updated by the weighted pattern of influence ties.

A second ingredient is the \emph{geometry of outcomes}.
When outcomes are discrete and structured---as with preference rankings---it is natural to represent the set of admissible states as a graph:
two rankings are neighbors if they differ by an elementary change (e.g., splitting/merging an indifference class).
This is the perspective we adopt: each node's state is a \emph{weak order} over $m$ alternatives, i.e., a ranking that allows ties.
The set of weak orders $\Omega(m)$ is finite but highly structured; it is canonically organized by the refinement (coarsening) partial order on weak orders and its cover graph $H$,
which encodes the elementary preference moves.
This combinatorial structure also aligns with classical social-choice foundations in which individual preferences are modeled as independent orderings over alternatives
\citep{arrow1951,sen1970,mascolell1995,schofield2003}.

Our interaction rule couples these two structures---a directed weighted influence network $G$ and the move graph $H$ on preferences---through a Borda-type aggregation step.
At each node $i$, neighbors' current rankings are mapped to score vectors and averaged using the row-stochastic weights,
producing a local Borda score vector that is then projected back to a weak order (with ties arising on score-equality boundaries)
\citep{borda1781,saari1995,terzopoulouendriss2021}.
Compared with scalar averaging, this creates three distinctive sources of nontrivial dynamics:
(i) the state space is finite and non-Euclidean; (ii) aggregation occurs in score space and returns to order space via a projection that can be discontinuous on tie boundaries;
and (iii) updates are \emph{bounded} by $H$: agents revise preferences incrementally, advancing at most one step per round toward the local aggregate ranking, which models deliberative friction and institutional constraints common in political and committee processes.

The boundedness assumption is a behavioral and institutional constraint rather than a technical trick:
it represents friction, limited attention, incremental deliberation, and stepwise revision of positions, common in political choice and committee processes.
Formally, boundedness is implemented by a \emph{bounded-step rule} that moves at most one edge in $H$ toward the current local target at each time step.
This gives a discrete-time dynamical system on $\Omega(m)^{\V}$ driven jointly by topology and weights.
The oscillatory phenomena studied in this paper---self-oscillations and forced oscillations---are robust features of the Borda-aggregation coupling: they are driven by the topology of the influence network $G$ and the geometry of the preference space $H$, and persist whether the update step is bounded or unbounded, as established in Proposition~\ref{prop:unbounded_fixed_periodic}.

This networked preference-dynamics perspective also differs sharply from the canonical ``static'' posture of social choice,
where aggregation is typically conceived as a one-shot mapping from a profile of individual preferences to a collective outcome and is assumed to be free of influence, persuasion, or manipulation.
Here, aggregation is \emph{endogenized} as a dynamical process sustained by social influence:
the eventual configuration, if it exists, is an emergent outcome of interaction and need not coincide with any individual's initial ranking.
Accordingly, equilibrium states have natural interpretations.
Homogeneous equilibria correspond to consensus outcomes (a network-mediated convergence of rankings),
while heterogeneous fixed points correspond to pluralist stable configurations whose existence depends on network structure, weights, and boundedness.

Most importantly for political and social interpretation, the model can exhibit \emph{oscillations}.
Even in deterministic dynamics, interaction may fail to settle, cycling through preference configurations.
We distinguish two conceptually different sources of periodic behavior.
First, in the absence of persistent sources, \emph{self-oscillations} can arise endogenously from the initial profile together with parity/topological mechanisms mediated by $G$ and $H$.
Second, when some nodes are \emph{persistent} (stubborn) sources---a standard modeling device in theories of social influence and polarization---consistent with heterogeneous dispositions of influence susceptibility
\citep{friedkin2011,acemoglu2011}, their fixed rankings can act as exogenous anchors representing parties, elites, media, or institutions.
When persistent nodes form \emph{contrarian camps} pinned to antipodal rankings in $\Omega(m)$, they can generate \emph{forced oscillations} among the remaining nodes, even when self-oscillations are absent.

The methodological contribution of this paper is therefore a mathematically tractable bridge between
network influence dynamics and the discrete geometry of preference orderings.
Our main results are stated directly in terms of directed topology, row-stochastic weights (random-walk normalization), and the bounded-step rule on the preference move graph.
We develop rigorous results on (a) existence and multiplicity of fixed points; and (b) mechanisms producing nontrivial periodic behavior,
with emphasis on (i) self-oscillations without persistence under synchronous updating (Variant~S) and (ii) forced oscillations induced by contrarian persistent camps.
As a contrast, we also analyze asynchronous updating (Variant~A), where parity mechanisms and periodicity behave differently.

It is useful to distinguish between two levels of results developed in this paper. 
Some properties of the dynamics are \emph{generic}, in the sense that they follow 
from the finite state space and the deterministic update rule and therefore hold 
for broad classes of networks and weight configurations (for example, eventual 
periodicity). Other results depend on \emph{specific structural configurations} 
of the influence network or of the preference move graph, such as directed cycles 
or bipartite topology, which create mechanisms capable of sustaining oscillations. 
Accordingly, the results below separate general dynamical properties of the model 
from oscillatory mechanisms that rely on particular network structures.

\section{Model}
\label{sec:model}

\subsection{State space}
\label{sec:state-space}
Let $\Alts=\{1,\dots,m\}$ be a finite set of $m > 1$ alternatives.\footnote{Sometimes,
when $m = 3$, we denote $\Alts=\{x, y, z\}$ and, when $m = 4$, $\Alts=\{x, y, z, u\}$.}
Let $\Omega(m)$ denote the finite set of all weak orders (total preorders, i.e.,
complete, transitive and reflexive relations) on $\Alts$, and write $\Omega:=\Omega(m)$.
In the terminology of social choice theory, individual preferences are typically
represented as complete and transitive orderings of alternatives (weak orders),
a standard convention in the literature \citep{sen1970}.
$\Omega$'s cardinality is the \emph{ordered Bell} (or \emph{Fubini}) number
\begin{equation}
\label{eq:fubini}
|\Omega(m)| \,=\, F(m)\,=\,\sum_{k=1}^{m} k!\,S(m,k),
\end{equation}
where $S(m,k)$ are Stirling numbers of the second kind \citep{stanleyEC1}.
In particular, $|\Omega(3)|=F(3)=13$ and $|\Omega(4)|=F(4)=75$.
We partially order $\Omega(m)$ by refinement: $\omega\le\omega'$ if $\omega'$ is obtained from $\omega$ by merging consecutive indifference classes.
This poset is not a lattice, since, e.g., two opposite strict orders have no common refinement in $\Omega(m)$; adjoining a least element turns it into the face lattice of the permutohedron.
Figure~\ref{fig:omega-lattices-m3m4} shows the Hasse (cover) graphs of the
weak-order refinement posets for $m=3$ and $m=4$; in this manuscript we denote this
Hasse (cover) graph by $H=(\Omega,E_H)$, a connected
graph whose edges correspond to elementary changes between adjacent weak orders (cover relations in the poset). We use $d_H$ to denote the shortest-path distance on $H$ (with unit edge lengths).

\begin{figure}
\figrule
\centering
\includegraphics{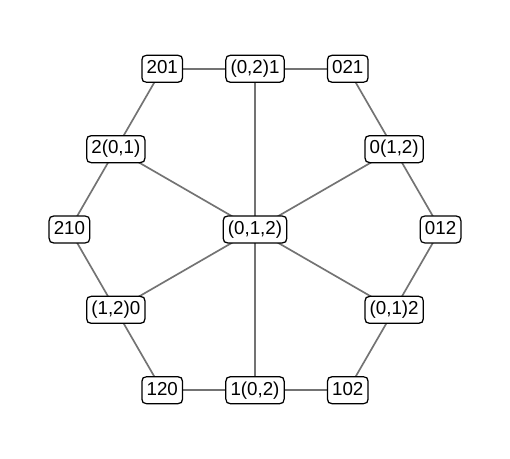}\\[4pt]
{\small (a) $\Omega(3)$: 13 weak orders, 18 cover edges}\\[10pt]
\includegraphics{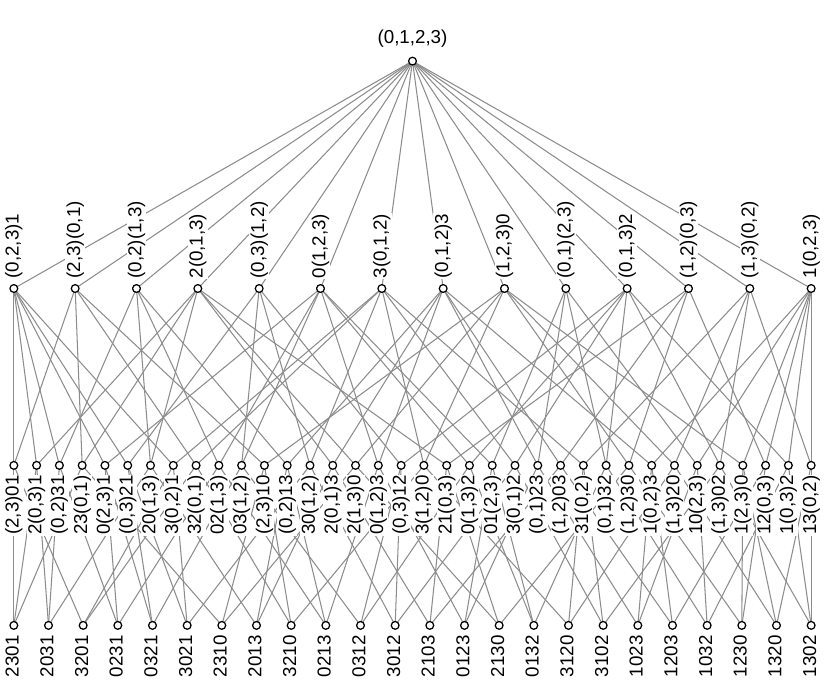}\\[4pt]
{\small (b) $\Omega(4)$: 75 weak orders, 158 cover edges}
\caption{Hasse diagrams (cover graphs) of the weak-order refinement posets $\Omega(3)$ and $\Omega(4)$. Each weak order is written as an ordered partition, best class first; e.g.\ $(0,1)2$ means $0\sim 1\succ 2$. Edges join weak orders that differ by merging two adjacent indifference classes. $|\Omega(3)|=F(3)=13$ and $|\Omega(4)|=F(4)=75$, where $F(m)$ is the ordered Bell (Fubini) number.}
\label{fig:omega-lattices-m3m4}
\figrule
\end{figure}

\begin{rem}[The Kemeny/Kendall distance]
For interpretive purposes, one may equip $\Omega$ with a Kemeny/Kendall metric
(based on pairwise disagreements) \citep{kendall1938,kemenysnell1962,bossertstorcken1992}.
For instance, when $m=3$ the fully indifferent weak order $(xyz)$ disagrees with
any strict ranking on all three pairwise comparisons, hence its Kemeny distance
to any strict order equals $3$.
This metric geometry does not by itself prescribe the update rules to be
introduced in Section~\ref{sec:update_rules}.
\end{rem}

\subsection{Borda scores and local profile aggregation}
\label{sec:borda-aggregation}
A Borda score map $b:\Omega\to\R^m$ assigns each weak order a score vector.
We use the averaged Borda convention for ties: if an indifference class occupies
consecutive rank positions, each alternative receives the average of the
corresponding strict Borda scores \citep{borda1781,saari1995,terzopoulouendriss2021}.
Given a score vector $s\in\R^m$, let $T(s)\in\Omega$ be the weak order obtained
by sorting alternatives by decreasing score, with ties in $s$ mapped to
indifference classes.
Thus $T$ is constant on polyhedral cones and discontinuous across tie hyperplanes.

A \emph{profile} is an assignment $\sigma=(\sigma_1,\dots,\sigma_n)\in \Omega^{\V}$
of a weak order to each node.
Given a profile $\sigma$ and a weight matrix $W$ (to be defined in
Section~\ref{sec:social-graph}), node $i$ forms the aggregated Borda score
\begin{equation}
\label{eq:agg_scores}
s_i(\sigma)=\sum_{j\in \V} w_{ij}\, b(\sigma_j)\in\R^m,
\end{equation}
and sets its (deterministic) Borda target to
\begin{equation}
\label{eq:target}
\tau_i(\sigma)=T\bigl(s_i(\sigma)\bigr)\in \Omega.
\end{equation}

\subsection{Social graph and weights}
\label{sec:social-graph}
Let $G=(\V,\E)$ be a finite directed graph with node set $\V=\{1,\dots,n\}$.
Each node $i\in\V$ carries a \emph{preference state} $\sigma_i\in\Omega$,
so the nodes of the social graph $G$ are in bijection with the nodes of the
preference state space $\Omega^{\V}$: the social network specifies \emph{who
influences whom}, while the preference space $\Omega$ specifies \emph{what
each node holds}.
A nonnegative weight matrix $W=(w_{ij})$ is \emph{compatible} with $G$ if
$w_{ij}>0$ only when $(i,j)\in\E$.
Throughout we assume \emph{row-stochasticity}:
\begin{equation}
\label{eq:row_stochastic}
\sum_{j\in \V} w_{ij}=1,\qquad i\in\V,
\end{equation}
so each node forms a random-walk-type convex combination of its in-neighbors'
Borda scores.
The undirected/unweighted case is recovered by taking $G$ undirected and setting
$W=D^{-1}A$, the random-walk matrix of $G$.

In a directed graph, influence flows along incoming edges.
For each node $i$, its \emph{in-neighborhood} is
\[
N_i^- := \{\, j\in V : w_{ij} > 0 \,\},
\]
so the local aggregation \eqref{eq:agg_scores} uses only the states of
in-neighbors: $s_i(\sigma) = \sum_{j\in N_i^-} w_{ij}\, b(\sigma_j)$.

\subsection{Persistent nodes}
\label{sec:persistent-nodes}

We study a discrete-time dynamical system in which each node $i\in V$ holds a
preference state $\sigma_i(t)\in\Omega$ at time $t\in\{0,1,2,\dots\}$ and
updates it at each step according to an interaction rule depending on the
states of its in-neighbors in $G$.

Before specifying the update rules, we introduce a structural partition of the
node set $V$ that is fixed throughout the dynamics: some nodes are
\emph{persistent} (their states are externally fixed for all time) and the
remaining nodes are \emph{nonpersistent} (they update according to the
interaction rule).

\begin{defn}[Persistent (boundary) nodes]
\label{defn:persistent}
Fix a set $P\subseteq V$ and an assignment
$\bar\sigma_P=(\bar\sigma_p)_{p\in P}\in \Omega(m)^P$.
Nodes in $P$ are \emph{persistent} (or \emph{boundary}) if their states are
kept constant and pinned for all times:
\[
\sigma_p(t)=\bar\sigma_p, \mbox{ for all }p\in P\mbox{ and all }t\ge 0.
\]
The remaining nodes $U:=V\setminus P$ are \emph{nonpersistent} (or \emph{free})
and evolve by the update rule.
\end{defn}

The partition $V=P\sqcup U$ is a structural property of the model, \emph{fixed before
any dynamics are specified}.
When $P=\varnothing$ all nodes are free; when $P\neq\varnothing$ the persistent
nodes act as exogenous anchors representing, for example, stubborn actors,
institutional positions, or media sources.

\subsection{Update interaction rules}
\label{sec:update_rules}
The most natural update rule is for each nonpersistent node to jump directly to
its local Borda target. This defines the \emph{unbounded} (or \emph{direct-jump})
step and serves as the natural benchmark against which the bounded step is compared.

\begin{defn}[Unbounded step / direct-jump map]
\label{defn:unbounded_step}
For a nonpersistent node $i\in U$, given its current state $\alpha\in\Omega$ and
its local Borda target $\beta=\tau_i(\sigma)\in\Omega$, define
\[
  \textsf{Jump}(\alpha,\beta) \;:=\; \beta.
\]
That is, the nonpersistent node moves directly to its local Borda target in a
single step, regardless of the distance $d_H(\alpha,\beta)$.
\end{defn}

In many social, political, and institutional settings, such unconstrained jumps
are implausible: agents revise their preferences incrementally, reflecting
informational, cognitive, or institutional constraints.
This motivates restricting each update to a single elementary move in $H$.

\begin{defn}[Bounded step]
\label{defn:bounded_step}
For a nonpersistent node $i\in U$, given its current state $\alpha\in\Omega$ and
its local Borda target $\beta=\tau_i(\sigma)\in\Omega$, define $\textsf{Step}(\alpha,\beta)$
as follows: if $\alpha=\beta$ return $\alpha$; otherwise move one step along a
shortest path in $H$ from $\alpha$ toward $\beta$.
To keep the dynamics deterministic, we assume a fixed tie-breaking rule when
multiple shortest-path neighbors exist (and we allow ``no-move'' clauses in
special ambiguous cases, as in certain implementation conventions).
\end{defn}

The bounded-step rule admits a natural interpretation that is both behavioral and
technical. From a behavioral perspective, it models incremental preference
revision: rather than instantaneously adopting a target ranking, agents adjust
their preferences through elementary changes in the space of orderings, consistent
with classical treatments of preference geometry
\citep{kemenysnell1962,saari1995} and models of bounded rationality
\citep{simon1957}.
From a technical perspective, boundedness restricts the admissible transitions in
$H$, ensuring that updates proceed through local moves.
This locality aligns with standard models of iterative opinion dynamics on
networks \citep{degroot1974,friedkin1990}.
The role of the bounded-step constraint relative to the unbounded benchmark is
examined throughout the paper, beginning in Section~\ref{sec:fixed_points}.

\subsection{Update variants}
\label{sec:variants}
Each step rule ($\textsf{Jump}$ or $\textsf{Step}$) combines with synchronous or
asynchronous scheduling, yielding four variants.
The two bounded variants are the primary objects of study; the two unbounded
variants serve as a benchmark for comparison.

\begin{defn}[Variant S: synchronous bounded Borda dynamics]
\label{defn:variantS}
At each time $t$, every nonpersistent node $i\in U$ updates simultaneously using
$\textsf{Step}$:
\[
\sigma_i(t+1)=\textsf{Step}\bigl(\sigma_i(t),\,\tau_i(\sigma(t))\bigr),
\qquad
\sigma_p(t+1)=\sigma_p(t)\mbox{ for }p\in P.
\]
We denote the induced map on $\Omega^U$ by $F_P$.
\end{defn}

\begin{defn}[Variant A: asynchronous bounded Borda dynamics]
\label{defn:variantA}
Fix an update schedule that selects a single nonpersistent node $I_t\in U$ at
each time $t$ (e.g., uniformly at random).
Only that node updates using $\textsf{Step}$:
\[
\sigma_{I_t}(t+1)=\textsf{Step}\bigl(\sigma_{I_t}(t),\,\tau_{I_t}(\sigma(t))\bigr),
\qquad
\sigma_{i}(t+1)=\sigma_i(t)\ \mbox{for }i\neq I_t.
\]
\end{defn}

\begin{defn}[Variant S$^\infty$: synchronous unbounded Borda dynamics]
\label{defn:variantSinf}
At each time $t$, every nonpersistent node $i\in U$ updates simultaneously using
$\textsf{Jump}$:
\[
\sigma_i(t+1)=\textsf{Jump}\bigl(\sigma_i(t),\,\tau_i(\sigma(t))\bigr)
=\tau_i(\sigma(t)),
\qquad
\sigma_p(t+1)=\sigma_p(t)\mbox{ for }p\in P.
\]
We denote the induced map on $\Omega^U$ by $F_P^{\infty}$.
\end{defn}

\begin{defn}[Variant A$^\infty$: asynchronous unbounded Borda dynamics]
\label{defn:variantAinf}
Fix an update schedule that selects a single nonpersistent node $I_t\in U$ at
each time $t$.
Only that node updates using $\textsf{Jump}$:
\[
\sigma_{I_t}(t+1)=\textsf{Jump}\bigl(\sigma_{I_t}(t),\,\tau_{I_t}(\sigma(t))\bigr)
=\tau_{I_t}(\sigma(t)),
\qquad
\sigma_{i}(t+1)=\sigma_i(t)\ \mbox{for }i\neq I_t.
\]
\end{defn}

\begin{rem}[Bounded versus unbounded dynamics]
\label{rem:bounded_vs_unbounded}
All four variants share the same aggregation step: the local Borda target
$\tau_i(\sigma)$ is computed identically via \eqref{eq:agg_scores}--\eqref{eq:target}
in every case, using both nonpersistent and persistent neighbours' current states.
The variants differ only in how far a nonpersistent node moves toward its target
in one round, and in whether all nonpersistent nodes update simultaneously or one
at a time.
The fixed-point sets of $F_P$ and $F_P^{\infty}$ coincide: a profile $\sigma$ is
a fixed point of either map if and only if $\tau_i(\sigma)=\sigma_i$ for all
$i\in U$.
The two maps produce the same eventual periodic orbits under the constructions
of Theorems~\ref{thm:self_traveling_wave} and~\ref{thm:period2_forcing},
as established in Proposition~\ref{prop:unbounded_fixed_periodic}.
The transient dynamics can differ: under Variants~S$^\infty$ and~A$^\infty$
the influence of persistent nodes is felt instantaneously rather than
incrementally.
Variants~S$^\infty$ and~A$^\infty$ are used as a benchmark in
Sections~\ref{sec:self_osc} and~\ref{sec:forced_osc}.
\end{rem}

\section{Fixed points and eventual periodicity} 
\label{sec:fixed_points_section}

\subsection{Eventual periodicity and fixed points}
\label{sec:fixed_points}

For a fixed persistent boundary condition and deterministic tie-breaking in $\textsf{Step}$, Variant~S defines a deterministic map
\(
F_P:\Omega^U\to\Omega^U.
\)

\begin{prop}[Eventual periodicity]
\label{prop:eventual_periodicity}
Every trajectory of Variant~S (with fixed $P$ and deterministic $\textsf{Step}$) is eventually periodic: there exist $\mu\ge 0$ and $p\ge 1$ such that $\sigma(t+p)=\sigma(t)$ for all $t\ge \mu$.
\end{prop}

\begin{proof}
$F_P$ is a deterministic map on the finite set $\Omega^U$, so every orbit enters a directed cycle after a finite transient.
\end{proof}

\begin{prop}[Fixed points and periodic orbits of the unbounded map]
\label{prop:unbounded_fixed_periodic}
Consider Variant~S$^\infty$ (Definition~\ref{defn:variantSinf}) with fixed
persistent boundary data $(P,\bar\sigma_P)$ and deterministic tie-breaking.

\begin{description}
  \item[(a) \textbf{Eventual periodicity}:] 
  Every trajectory of Variant~S$^\infty$ is eventually periodic.

  \item[(b) \textbf{Fixed points}:] 
  The fixed-point sets of $F_P$ and $F_P^{\infty}$ coincide:
  a profile $\sigma\in\Omega^V$ is a fixed point of $F_P^{\infty}$ if and only if
  $\tau_i(\sigma)=\sigma_i$ for all $i\in U$.
  In particular, every fixed point of Variant~S is a fixed point of
  Variant~S$^\infty$, and vice versa.

  \item[(c) \textbf{Self-oscillations persist}:] 
  Under the hypotheses of Theorem~\ref{thm:self_traveling_wave}
(directed $\ell$-cycle $C$ in $G$, $k$-cycle in $H$ with $k\mid\ell$,
targeting condition $\tau_i(\sigma)=\sigma_{i-1}$, and initialisation
$\sigma_i(0)=\omega_{i\bmod k}$), the restriction $(\sigma_i(t))_{i\in C}$ is periodic with period $k$ under Variant~S$^\infty$ as well. That is, the self-oscillation of Theorem~\ref{thm:self_traveling_wave} survives unchanged when the bounded step is replaced by the direct-jump map.

  \item[(d) \textbf{Forced oscillations persist}:]
  Under the hypotheses of Theorem~\ref{thm:period2_forcing},
  Variant~S$^\infty$ has the period--$2$ orbit
  \[
  (\rho,\rho^\dagger)\longleftrightarrow(\rho^\dagger,\rho).
  \]
  Along this Variant~S$^\infty$ orbit, each free node's aggregate score
  has $\ell^\infty$-distance at least $(1-2\varepsilon)/2>0$ from the
  tie hyperplane.
\end{description}
\end{prop}

\begin{proof}
\textbf{(a)} $F_P^{\infty}$ is a deterministic map on the finite set $\Omega^U$,
so every orbit is eventually periodic.

\textbf{(b)} By definition,
$F_P^{\infty}(\sigma)_i = \textsf{Jump}(\sigma_i,\tau_i(\sigma)) = \tau_i(\sigma)$.
Hence $F_P^{\infty}(\sigma)=\sigma$ if and only if $\tau_i(\sigma)=\sigma_i$ for all
$i\in U$, which is the same condition as for $F_P$
(since $\textsf{Step}(\alpha,\alpha)=\alpha$ for all $\alpha\in\Omega$).

\textbf{(c)} Under the hypotheses of Theorem~\ref{thm:self_traveling_wave},
the targeting condition gives $\tau_i(\sigma(t))=\sigma_{i-1}(t)$ for all $i\in C$.
Under Variant~S$^\infty$:
\[
  \sigma_i(t+1) = \textsf{Jump}\bigl(\sigma_i(t),\,\tau_i(\sigma(t))\bigr)
  = \tau_i(\sigma(t)) = \sigma_{i-1}(t).
\]
By induction, $\sigma_i(t)=\omega_{(i-t)\bmod k}$ for all $i\in C$ and all $t\ge 0$.
Since the $\omega_j$ are distinct, the profile $(\sigma_i(t))_{i\in C}$ first returns
to its initial value when $t\equiv 0\pmod{k}$, so the period is exactly $k$.

We note that the proof of Theorem~\ref{thm:self_traveling_wave} under Variant~S
uses the condition $\textsf{Step}(\omega_r,\omega_{r-1})=\omega_{r-1}$, which holds
because adjacent states on the $H$-cycle are at distance $1$ in $H$.
Under Variant~S$^\infty$, the same conclusion holds without this condition:
$\textsf{Jump}(\omega_r,\omega_{r-1})=\omega_{r-1}$ for any $\omega_r,\omega_{r-1}\in\Omega$,
whether or not they are adjacent in $H$.
Thus Variant~S$^\infty$ requires \emph{fewer} hypotheses than Variant~S to produce
the same periodic orbit.

\textbf{(d)} Part~(a) of Theorem~\ref{thm:period2_forcing} gives the
stated Variant~S$^\infty$ orbit and its Variant~S$^\infty$-specific
margin from the tie hyperplane.

\end{proof}

\begin{rem}[Oscillations are robust to the step rule]
\label{rem:oscillations_robust}
Proposition~\ref{prop:unbounded_fixed_periodic} shows that the oscillatory
phenomena established in Theorems~\ref{thm:self_traveling_wave}
and~\ref{thm:period2_forcing} are \emph{not} artifacts of the bounded-step
rule: they persist, with the same periods, when the bound is removed entirely.
The driving mechanisms are the directed-cycle and bipartite topology of the
influence network $G$, the cycle and antipodal structure of the preference
graph $H$, and (for forced oscillations) the contrarian boundary conditions.
The bounded-step rule is a behaviorally motivated modeling choice describing
incremental preference revision; it neither creates nor destroys the oscillatory
behavior identified in this paper.

We also note a terminological distinction: our use of ``bounded'' refers to a
constraint on \emph{how far} a nonpersistent node moves in one update step,
not on \emph{which nodes interact}.
This differs from the bounded-confidence models of opinion dynamics
\citep{deffuant2000,hegselmann2002}, where ``bounded'' describes a threshold
on \emph{who influences whom}.
The two senses of boundedness are conceptually independent.
\end{rem}

\begin{prop}[Local Borda equilibria as fixed points]
\label{prop:local_equilibria_fixed}
Fix persistent boundary data $(P,\bar\sigma_P)$ and consider Variant~S.
If a profile $\sigma\in\Omega^\V$ satisfies 
\[
\tau_i(\sigma)=\sigma_i, \mbox{ for all } i\in U,
\]
then $\sigma_U$ is a fixed point of the induced map $F_P:\Omega^U\to\Omega^U$.
\end{prop}

\begin{proof}
For each free node $i\in U$, $\sigma_i=\tau_i(\sigma)$ implies
$\textsf{Step}(\sigma_i,\tau_i(\sigma))=\sigma_i$ by the definition of $\textsf{Step}$.
Persistent nodes do not update. Hence $F_P(\sigma_U)=\sigma_U$.
\end{proof}

\begin{lem}[Frozen targets imply finite-time stabilization]
\label{lem:frozen_targets}
Assume that for some $t_0$ the targets become time-independent, i.e.,
there exist $\bar\tau_i\in\Omega$ such that for all $t\ge t_0$ and all $i\in U$,
\[
\tau_i(\sigma(t))=\bar\tau_i.
\]
Then for each $i\in U$, the distance $d_H(\sigma_i(t),\bar\tau_i)$ is nonincreasing in $t$
and reaches $0$ in at most $d_H(\sigma_i(t_0),\bar\tau_i)$ steps.
In particular, if all $\bar\tau_i$ coincide with a single $\bar\omega$, then Variant~S reaches
consensus $\sigma_i\equiv \bar\omega$ within at most $\mathrm{diam}(H)$ steps after $t_0$.
\end{lem}

\begin{proof}
By construction, $\textsf{Step}(\alpha,\beta)$ moves one step along a shortest path in $H$
from $\alpha$ toward $\beta$ (or stays put if $\alpha=\beta$).
Hence $d_H(\sigma_i(t+1),\bar\tau_i)=\max\{d_H(\sigma_i(t),\bar\tau_i)-1,0\}$.
The consensus claim follows by taking the maximum over $i$ and using $\mathrm{diam}(H)$.
\end{proof}

\subsection{Consensus fixed points} 

\begin{prop}[Consensus profiles as fixed points]
\label{prop:consensus_fixed_points} 
Assume $P=\varnothing$ and $W$ is row-stochastic.
For any $\omega\in\Omega$, the consensus profile $\sigma_i\equiv \omega$ (for all $i \in V$) is a fixed point of Variant~S.
\end{prop}

\begin{proof}
If $\sigma_j=\omega$ for all $j$, then $s_i(\sigma)=\sum_j w_{ij} b(\omega)=b(\omega)$ by \eqref{eq:row_stochastic}.
Hence $\tau_i(\sigma)=T(b(\omega))=\omega$, and $\textsf{Step}(\omega,\omega)=\omega$ for all $i$.
\end{proof}

\begin{rem}[Multiplicity of fixed points (direction)]
Even without persistent (boundary) nodes, nonconsensus fixed points can arise, e.g., through ties in score space and the bounded-step restriction on $H$.
A natural problem is to characterize, in terms of graph topology and weights, when Variant~S admits only consensus equilibria versus additional equilibria.
In this manuscript we primarily develop sufficient conditions for the existence of nontrivial periodic orbits (self-oscillations and forced oscillations under contrarian boundary conditions).
\end{rem}

\subsection{Network dependence of convergence}

\begin{prop}[Network dependence of asymptotic convergence]
Under Variant S, asymptotic convergence of the dynamics depends
on the structure of the influence network $(G,W)$. In particular:

(i) If the trajectory enters a fixed point, then it converges.

(ii) If $(G,W)$ satisfies the conditions of Theorem~4.1, then
there exist initial conditions for which the dynamics do not
converge.

(iii) If $(G,W)$ satisfies the bipartite forcing conditions of
Theorem~5.2, then there exist initial conditions for which the
dynamics exhibit periodic orbits of even period and hence do not
converge.
\end{prop}

\begin{proof}
Part (i) follows from the definition of convergence.

Part (ii) follows from Theorem~4.1, which establishes the
existence of periodic orbits of period greater than one under
directed-cycle structures. Such trajectories are not eventually
constant and therefore do not converge.

Part (iii) follows from Theorem~5.2, which constructs periodic
orbits (of even period) under bipartite forcing configurations.
Again, these trajectories are not eventually constant and hence
do not converge.
\end{proof}

\section{Self-oscillations without persistence} 
\label{sec:self_osc}

Proposition~\ref{prop:eventual_periodicity} guarantees that any nonconvergent trajectory eventually cycles, but it does not explain \emph{why} cycles arise nor provide conditions ensuring a specific period.
We now record topology/weight/Step conditions that guarantee genuine oscillations ($p>1$) without any persistent nodes.

\subsection{Oscillations on directed cycles} 
The next theorem shows that if the social graph contains a directed cycle with essentially permutation-like influence, then periodic orbits exist whenever the move graph $H$ contains a cycle.

\begin{thm}[Self-oscillations via oscillations on a directed influence cycle] 
\label{thm:self_traveling_wave}
Assume $P=\varnothing$.  
Suppose the social graph contains a directed cycle $C=(1\to 2\to \cdots \to \ell \to 1)$ such that for every $i\in C$,
\[
w_{i,i-1}>0 \quad\mbox{(with $i-1$ understood mod $\ell$)}, \qquad w_{ij}=0 \mbox{ for } j\neq i-1,
\]
so each node copies the unique predecessor in the cycle at the \emph{score-aggregation} stage.
Assume moreover that for all profiles $\sigma$, the target at $i\in C$ satisfies
\(
\tau_i(\sigma)=\sigma_{i-1}.
\)
Let $(\omega_0,\omega_1,\dots,\omega_{k-1},\omega_0)$ be a simple cycle in the move graph $H$ (adjacent successive states).
Assume that $k\mid\ell$.
Initialize the nodes on $C$ by
\[
\sigma_i(0)=\omega_{i \bmod k}\qquad (i\in C),
\]
and choose arbitrary initial states off $C$.
If $\textsf{Step}(\omega_r,\omega_{r-1})=\omega_{r-1}$ for all $r$ (i.e., one bounded step reaches the predecessor on the $H$-cycle),
then under Variant~S the restriction $(\sigma_i(t))_{i\in C}$ is periodic with period $k$.
In particular, with no persistent nodes the nodes on $C$ perform $k$-oscillations; the full trajectory is eventually periodic with period divisible by $k$, and if $\V=C$ the full dynamics has a $k$-cycle.
\end{thm}

\begin{proof}
We prove by induction on $t$ that
\[
\sigma_i(t)=\omega_{(i-t)\bmod k}
\qquad\text{for every }i\in C.
\]
For $t=0$, this is the prescribed initialization. If the formula
holds at time $t$, then the targeting condition gives
\[
\tau_i(\sigma(t))=\sigma_{i-1}(t)
=\omega_{(i-1-t)\bmod k}
=\omega_{(i-(t+1))\bmod k}.
\]
Since successive states on the $H$-cycle are adjacent,
\[
\sigma_i(t+1)
=\textsf{Step}\bigl(
\omega_{(i-t)\bmod k},
\omega_{(i-(t+1))\bmod k}
\bigr)
=\omega_{(i-(t+1))\bmod k}.
\]
Thus the formula holds for all $t\ge0$. Since the $\omega_r$ are
distinct, the restriction $(\sigma_i(t))_{i\in C}$ has period $k$.
\end{proof}

\begin{rem}[Targeting condition: score-and-weight derivation]
\label{rem:targeting_condition}
The targeting condition $\tau_i(\sigma)=\sigma_{i-1}$ in
Theorem~\ref{thm:self_traveling_wave} is a consequence of the weight
structure, not an independent assumption.
Under the hypothesis $w_{i,i-1}>0$ and $w_{ij}=0$ for $j\neq i-1$,
row-stochasticity forces $w_{i,i-1}=1$, so the aggregated score is
\[
  s_i(\sigma)
  = \sum_{j\in V} w_{ij}\,b(\sigma_j)
  = 1\cdot b(\sigma_{i-1})
  = b(\sigma_{i-1}).
\]
Since $T\circ b=\mathrm{id}$ on $\Omega$ (the projection $T$ recovers
the ordering from its own Borda score vector), we get
$\tau_i(\sigma)=T(s_i(\sigma))=T(b(\sigma_{i-1}))=\sigma_{i-1}$
for every profile $\sigma$.
The targeting condition is therefore implied by the hypotheses of
Theorem~\ref{thm:self_traveling_wave}.

More generally, if $w_{i,i-1}<1$ (other neighbors have positive weight),
the targeting condition holds if and only if the perturbed score
$s_i(\sigma)=b(\sigma_{i-1})+\delta_i(\sigma)$, where
\[
  \delta_i(\sigma)
  := \sum_{j\neq i-1} w_{ij}\bigl(b(\sigma_j)-b(\sigma_{i-1})\bigr),
\]
lies in the same polyhedral cone of $T$ as $b(\sigma_{i-1})$.
A sufficient condition is
\[
  \lVert\delta_i(\sigma)\rVert_\infty
  \;<\;
  d_\infty\!\bigl(b(\sigma_{i-1}),\,\partial\mathcal{C}(\sigma_{i-1})\bigr),
\]
where $\mathcal{C}(\sigma_{i-1})$ is the polyhedral cone of $T$
containing $b(\sigma_{i-1})$ and $d_\infty$ is the $\ell^\infty$
distance to its boundary.
Since $\lVert b(\sigma_j)-b(\sigma_{i-1})\rVert_\infty\le m-1$ for all
$\sigma_j,\sigma_{i-1}\in\Omega$, a uniform sufficient condition is
\[
  (1-w_{i,i-1})(m-1)
  \;<\;
  d_\infty\!\bigl(b(\sigma_{i-1}),\,\partial\mathcal{C}(\sigma_{i-1})\bigr).
\]
This inequality is satisfied whenever the predecessor's weight
$w_{i,i-1}$ is sufficiently close to $1$ relative to the margin
of $\sigma_{i-1}$ from the nearest score-tie hyperplane.
\end{rem}

\begin{cor}[Existence of nontrivial oscillations for $m=3$]
\label{cor:m3_osc}
For $m=3$, the move graph $H$ in Figure~\ref{fig:omega-lattices-m3m4} contains the $12$-cycle formed by the boundary of the hexagon, in which the six strict rankings alternate with the six weak orders having exactly two indifference classes,
\begin{gather*}
012 \to (0,1)2 \to 102 \to 1(0,2) \to 120 \to (1,2)0 \to 210 \to 2(0,1) \\
\to 201 \to (0,2)1 \to 021 \to 0(1,2) \to 012,
\end{gather*}
as well as $4$-cycles through the fully indifferent weak order, e.g.\ $012 \to (0,1)2 \to (0,1,2) \to 0(1,2) \to 012$.
Hence Theorem~\ref{thm:self_traveling_wave} yields nontrivial oscillations
of the nodes on the cycle, of period $12$ (resp.\ $4$), whenever $P=\varnothing$ and the
social graph contains a directed influence cycle of length $\ell$
divisible by $12$ (resp.\ $4$) satisfying its weight/target hypotheses.
\end{cor} 

Figure~\ref{fig:selfosc_cycle} displays a self oscillation as in Theorem~\ref{thm:self_traveling_wave}.

\begin{figure}[h!]
\centering
\includegraphics{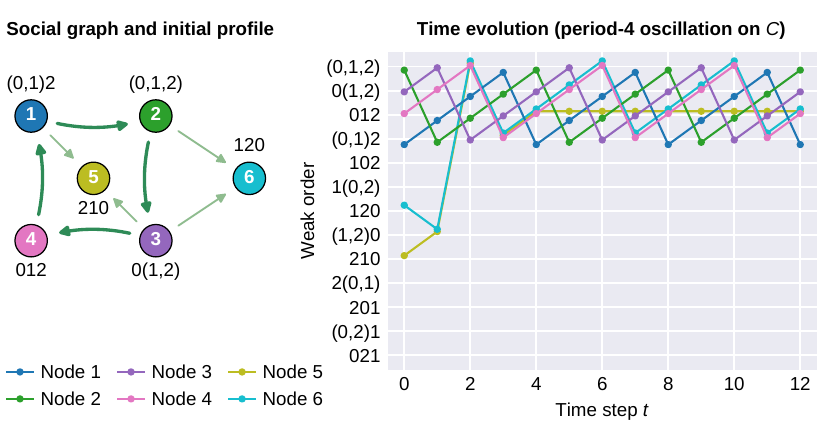}
\caption{A self-oscillation as in Theorem~\ref{thm:self_traveling_wave} for $m=3$ under Variant~S with $P=\varnothing$. Left: the social graph (arrows $j\to i$ indicate that $j$ influences $i$) and the initial profile. Nodes $1\to2\to3\to4\to1$ form a directed cycle $C$ ($\ell=4$) with $w_{i,i-1}=1$, initialised along the $4$-cycle $012\to(0,1)2\to(0,1,2)\to0(1,2)$ of $H$ ($k=4$). Node~$5$ ($w_{51}=w_{53}=\tfrac12$) and node~$6$ ($w_{62}=w_{63}=\tfrac12$) lie off $C$. Right: the nodes on $C$ oscillate with period $4$, as predicted. Node~$5$ settles at $012$, while node~$6$ is drawn into the same period-$4$ cycle. Curves are slightly offset vertically for visibility.}
\label{fig:selfosc_cycle}
\end{figure}

\subsection{Even-period lifting} 
Bipartite social topology naturally induces a two-step description of synchronous dynamics and helps explain the prevalence of even-period cycles.

\begin{lem}[Two-step decoupling on bipartite graphs]
\label{lem:bipartite_decoupling}
Suppose $G$ is bipartite with $\V=A\sqcup B$ and influence is across the cut only, i.e.,
$w_{ij}=0$ whenever $i,j\in A$ or $i,j\in B$, 
where $A\sqcup B$ denotes a disjoint union (so $A\cap B=\varnothing$ and $A\cup B=\V$).
Write $\sigma=(\sigma_A,\sigma_B)$.
Then the synchronous map $F:\Omega^{\V}\to\Omega^{\V}$ can be written as
\[
F(\sigma_A,\sigma_B)=\bigl(F_A(\sigma_B),\,F_B(\sigma_A)\bigr)
\]
for suitable maps $F_A:\Omega^{B}\to\Omega^{A}$ and $F_B:\Omega^{A}\to\Omega^{B}$, and hence
\[
F^{2}(\sigma_A,\sigma_B)=\bigl(G_A(\sigma_A),\,G_B(\sigma_B)\bigr),
\qquad
G_A:=F_A\circ F_B,\ \ G_B:=F_B\circ F_A.
\]
\end{lem}

\begin{proof}
Bipartiteness implies that each node in $A$ depends only on states in $B$ (and vice versa), so the one-step update on $A$ defines $F_A(\sigma_B)$ and similarly $F_B(\sigma_A)$.
Composing twice yields the stated decoupling for $F^2$.
\end{proof}

\begin{thm}[Even-period lifting on bipartite graphs]
\label{thm:even_period_lift}
Under the hypotheses of Lemma~\ref{lem:bipartite_decoupling}, if the composed map $G_A=F_A\circ F_B$ has a $k$-cycle on $\Omega^A$, then $F$ has a periodic orbit on $\Omega^{\V}$ that returns after $2k$ steps.
Its exact period is $2k$ if $k$ is even, and $k$ or $2k$ if $k$ is odd.
\end{thm}

\begin{proof}
Let $\sigma_A(0)$ lie on a $k$-cycle of $G_A$ and define $\sigma_B(0):=F_B(\sigma_A(0))$.
Then $(\sigma_A(2t),\sigma_B(2t))$ evolves by $(G_A^t(\sigma_A(0)),G_B^t(\sigma_B(0)))$; since $G_B^k F_B=F_B G_A^k$, it repeats after $k$ even steps, i.e.\ $F^{2k}(\sigma(0))=\sigma(0)$.
Let $p$ be the exact period, so $p\mid 2k$.
If $p$ is even, then $G_A^{p/2}(\sigma_A(0))=\sigma_A(0)$, so $k\mid p/2$ and $p=2k$.
If $p$ is odd, then $F^{2p}(\sigma(0))=\sigma(0)$ gives $G_A^{p}(\sigma_A(0))=\sigma_A(0)$, so $k\mid p$; hence $p=k$, which forces $k$ to be odd.
(For $k=1$ the orbit may be a fixed point of $F$.)
\end{proof}

\section{Forced oscillations with contrarian camps} 
\label{sec:forced_osc}

We now incorporate persistent sources and analyze forced oscillations, with emphasis on contrarian camps (persistent rankings partitioned into two antipodal camps in $\Omega$).

\subsection{Induced dynamics on free nodes}
Fix a set of persistent nodes $P\subset \V$ and boundary states $\bar\sigma_P$.
Variant~S induces a deterministic map $F_P:\Omega^U\to \Omega^U$ on the free nodes $U$.
When $F_P$ has no fixed point, oscillations are unavoidable.

\begin{thm}[No fixed point implies forced oscillation]
\label{thm:nofix_forced_osc}
If $F_P$ has no fixed point, then every trajectory of Variant~S on $U$ is eventually periodic with period $p>1$.
\end{thm}

\begin{proof}
By Proposition~\ref{prop:eventual_periodicity}, every orbit enters a directed cycle.
If there are no fixed points, the cycle length cannot be $1$.
\end{proof}

Thus, a key task is to identify topological/weight conditions (and contrarian boundary conditions) that \emph{preclude} fixed points or that explicitly construct periodic orbits.

\subsection{Antipodes on the preference state space}

\begin{defn}[Antipodes]
An \emph{antipode} involution $\dagger:\Omega\to\Omega$ in the state space $\Omega$ represents a complete reversal of pairwise comparisons (for strict rankings, reverse the order; for weak orders, reverse the ordered partition).
\end{defn}

For averaged Borda scores, antipodes satisfy the score antisymmetry
\begin{equation}
\label{eq:borda_antipode}
b(\omega^\dagger)=(m-1)\one - b(\omega),
\end{equation}
so after centering scores (subtracting $(m-1)\one/2$) antipodes negate.

\begin{lem}[Antipodes exist on $\Omega(m)$ and are $H$-automorphisms]
\label{lem:antipodes_exist}
Represent $\omega\in\Omega(m)$ as an ordered partition
$(C_1\succ C_2\succ \cdots \succ C_k)$ of $\Alts=\{1,\dots,m\}$ into indifference classes.
Define the \emph{antipode} map $\dagger:\Omega(m)\to\Omega(m)$ by reversal,
\[
\omega^\dagger := (C_k\succ C_{k-1}\succ \cdots \succ C_1).
\]
Then $\dagger$ is an involution, i.e.\ $(\omega^\dagger)^\dagger=\omega$, and it preserves cover relations
in the refinement poset, hence induces an automorphism of its cover graph $H$.
Moreover, averaged Borda scores satisfy
\[
b(\omega^\dagger)=(m-1)\mathbf{1}-b(\omega).
\]
\end{lem}

\begin{proof}
Reversing an ordered partition yields another total preorder, and reversal is clearly an involution.
A cover relation in the refinement poset corresponds to splitting an indifference class into two consecutive classes;
reversing the order of classes sends such a split to a split, hence preserves cover relations and therefore adjacency in $H$.
For Borda scores with ties, each alternative's score is the average of the ranks it occupies.
Under reversal, rank $r$ becomes rank $(m-1-r)$, hence scores transform by $b\mapsto (m-1)\mathbf1-b$.
\end{proof}

\begin{defn}[Contrarian persistent camps]
A persistent boundary condition $\bar\sigma_P$ is \emph{contrarian} if $P$ can be partitioned as $P=P^+\sqcup P^-$ and there exists $\rho\in\Omega$ such that
$\bar\sigma_p=\rho$ for $p\in P^+$ and $\bar\sigma_p=\rho^\dagger$ for $p\in P^-$.
\end{defn}

\subsection{Alternating targets} 
The bounded-step rule implies that if a node's target alternates between two distinct states, the node cannot converge to a fixed point. Under a mild uniqueness assumption on shortest paths, one gets a clean two-cycle.

\begin{lem}[Forced local oscillation by alternating targets]
\label{lem:alt_targets}
Fix a node $i$ and suppose there exist $\beta\neq \gamma$ in $\Omega$ such that for all sufficiently large $t$,
\[
\tau_i(\sigma(t))=\beta \mbox{ for all even }t,\qquad \tau_i(\sigma(t))=\gamma \mbox{ for all odd }t.
\]
Assume the shortest path between $\beta$ and $\gamma$ in $H$ is unique and that $\textsf{Step}$ always moves to the unique neighbor on the corresponding geodesic.
Then $\sigma_i(t)$ is eventually periodic with period $2$ (it alternates between two adjacent states on that geodesic).
\end{lem}

\begin{proof}
Let $(v_0,\dots,v_L)$ be the unique geodesic from $v_0=\beta$ to $v_L=\gamma$.
For $t$ large enough, if $\sigma_i(t)=v_r$ and the target is $\beta=v_0$, the unique bounded step moves to $v_{r-1}$ when $r>0$ (and stays at $v_0$ when $r=0$).
If the target is $\gamma=v_L$, the unique bounded step moves to $v_{r+1}$ when $r<L$ (and stays at $v_L$ when $r=L$).
Thus, for all sufficiently large $t$, the index $r(t)$ alternates deterministically between $r$ and $r-1$ (or between $0$ and $1$, or between $L$ and $L-1$ at the boundaries), yielding a period-$2$ orbit in $\Omega$.
\end{proof}

Lemma~\ref{lem:alt_targets} reduces forced oscillation questions to identifying conditions that make targets alternate, which can be achieved by bipartite forcing and (in undirected settings) by spectral parity.

\subsection{Spectral parity for bipartite walks} 

We record a basic spectral fact for random walks on bipartite graphs and explain how it can be leveraged to force alternating score signals under contrarian camps.
This is where spectral theory enters most cleanly.

\begin{lem}[A $(-1)$ eigenmode for bipartite random walks]
\label{lem:-1_eigenmode}
Let $G$ be a connected undirected bipartite graph with partition $\V=A\sqcup B$ and random-walk matrix $W=D^{-1}A$ (row-stochastic).
Assume there are no self-loops.
Define $f\in\R^{\V}$ by $f_i=1$ for $i\in A$ and $f_i=-1$ for $i\in B$.
Then $W f=-f$, so $-1$ is an eigenvalue of $W$.
\end{lem}

\begin{proof}
If $i\in A$, all neighbors lie in $B$, so $(Wf)_i$ is the average of $f_j=-1$ over neighbors of $i$, hence $(Wf)_i=-1=-f_i$.
Similarly, if $i\in B$, neighbors lie in $A$ and $(Wf)_i=1=-f_i$.
\end{proof}

\begin{rem}[Period-$2$ parity and even cycles]
The presence of a $(-1)$ eigenvalue is the linear-algebraic signature of period-$2$ parity in bipartite random walks \citep{levinpereswilmer2009,chung1997}.
In our preference dynamics, Borda aggregation is linear in the score vectors, but projection back to orderings and bounded steps introduce nonlinearity.
Nevertheless, away from score-tie hyperplanes (so targets are locally constant), the dynamics is \emph{piecewise affine} in score space, and the $(-1)$ mode provides a robust mechanism for alternating target signals.
\end{rem}

\subsection{Reachability and even-period forcing} 

In a directed weighted influence graph, persistent nodes can generate \emph{forced} oscillations only in the portion of the network they can reach.
Write $P\subset V$ for persistent nodes and $U=V\setminus P$ for free nodes.
Let $D_W$ be the simple directed \emph{influence support digraph} on $V$ with an arc $j\to i$ whenever $w_{ij}>0$.
For a set $S\subseteq V$, let $\Reach(S)$ denote the set of vertices reachable from $S$ in $D_W$.

\begin{lem}[Irrelevance of unreachable persistence]
\label{lem:unreachable_persistence}
If $i\in U\setminus \Reach(P)$, then for every $t\ge 0$ the local aggregate score $s_i(t)$ depends only on the initial states on $\Reach(\{i\})$ (and not on the persistent boundary condition on $P$).
In particular, persistent camps confined to components that do not reach $i$ have no effect on the evolution of $i$.
\end{lem}

\begin{proof}
By definition of $D_W$, an arc $j\to i$ means that the state of $j$ enters $s_i(t)$ with positive weight at time $t$.
Iterating the update, the set of vertices that can influence $i$ by time $t$ is contained in the set of vertices that can reach $i$ in $D_W$ by a directed path of length at most $t$.
If no directed path from $P$ reaches $i$, then no persistent node can enter $s_i(t)$ at any time.
\end{proof}

We now isolate a clean topological/spectral mechanism that produces \emph{even-period} oscillations under a contrarian partition of persistent nodes into two antipodal camps.
Fix an antipodal pair $(\rho,\rho^\dagger)\in\Omega\times\Omega$ and partition the persistent nodes into two camps
\[
P^+\sqcup P^- = P,\qquad \sigma_p(t)\equiv \rho\ \ (p\in P^+),\qquad \sigma_p(t)\equiv \rho^\dagger\ \ (p\in P^-).
\]
Let $C\subseteq U$ be a \emph{closed communicating class} for the free-to-free kernel $W_{UU}$: for all $i\in C$ we have $\sum_{j\in U\setminus C} w_{ij}=0$ and the induced digraph $D_W[C]$ is strongly connected.
Assume moreover that $W_{CC}$ is \emph{periodic with period $2$} (imprimitive): equivalently, $C$ admits a cyclic decomposition $C=A\sqcup B$ such that $w_{ij}>0$ implies $(i,j)\in A\times B$ or $(i,j)\in B\times A$, and (under irreducibility, by Perron--Frobenius theory) such a decomposition exists if and only if $-\rho(W_{CC})\in\sigma(W_{CC})$, where $\rho(W_{CC})$ is the spectral radius \citep{seneta2006,levinpereswilmer2009}.
Because rows of $C$ may put weight on persistent nodes, $W_{CC}$ is in general substochastic, so $\rho(W_{CC})<1$ is possible and $-1$ need not be an eigenvalue; for instance, in Theorem~\ref{thm:period2_forcing} below, $W_{CC}$ has eigenvalues $\pm(1-\varepsilon)$.
In this case, the two-step kernel $W_{CC}^2$ leaves $A$ and $B$ invariant and defines aperiodic chains on each cyclic class.

The following theorem gives an explicit two-node construction of a
forced period--$2$ orbit.

\begin{thm}
\label{thm:period2_forcing}
Let $m=2$, let
$\rho=(a_1\succ a_2)$, $\rho^\dagger=(a_2\succ a_1)$, and let
$\eta=(a_1\sim a_2)$. Let $P^+=\{p\}$ and $P^-=\{q\}$ be
persistent camps, with
\[
\bar\sigma_p=\rho,\qquad \bar\sigma_q=\rho^\dagger.
\]
Let $C=\{i,j\}$ be two free nodes. For some
$0<\varepsilon<\tfrac12$, suppose
\[
w_{ij}=1-\varepsilon,\qquad w_{ip}=\varepsilon,
\qquad w_{ji}=1-\varepsilon,\qquad w_{jq}=\varepsilon,
\]
and that all other weights in the rows of $i$ and $j$ are zero.

Then the following hold.
\begin{description}
  \item[(a)] Under Variant~S$^\infty$, the initial state
  \[
  (\sigma_i(0),\sigma_j(0))=(\rho,\rho^\dagger)
  \]
  belongs to a period--$2$ orbit alternating with
  $(\rho^\dagger,\rho)$. Along this Variant~S$^\infty$ orbit, each
  free node's aggregate score has $\ell^\infty$-distance at least
  $(1-2\varepsilon)/2$ from the tie hyperplane.

  \item[(b)] Under Variant~S, the same initial state belongs to a
  period--$2$ orbit alternating with $(\eta,\eta)$. Along this
  Variant~S orbit, each free node's aggregate score has
  $\ell^\infty$-distance at least
  \[
  \frac{\min\{\varepsilon,1-2\varepsilon\}}{2}>0
  \]
  from the tie hyperplane.
\end{description}
\end{thm}

\begin{proof}
The Borda score vectors are
\[
b(\rho)=(1,0),\qquad b(\rho^\dagger)=(0,1),
\qquad b(\eta)=(\tfrac12,\tfrac12).
\]
At $(\sigma_i,\sigma_j)=(\rho,\rho^\dagger)$, the aggregate scores
are
\[
s_i=(\varepsilon,1-\varepsilon),
\qquad
s_j=(1-\varepsilon,\varepsilon).
\]
Thus $\tau_i=\rho^\dagger$ and $\tau_j=\rho$. Under Variant~S$^\infty$,
the next state is $(\rho^\dagger,\rho)$. At that state, the aggregate
scores are
\[
s_i=(1,0),
\qquad
s_j=(0,1),
\]
so the next state is $(\rho,\rho^\dagger)$. The score gaps at the two
states of this orbit are $1-2\varepsilon$ and $1$, respectively.
Therefore the aggregate scores remain at $\ell^\infty$-distance at
least $(1-2\varepsilon)/2$ from the tie hyperplane, proving part~(a).

In the weak-order cover graph for $m=2$, the unique shortest path from
$\rho$ to $\rho^\dagger$ is
\[
\rho\longrightarrow\eta\longrightarrow\rho^\dagger.
\]
Consequently, under Variant~S, one bounded step from
$(\rho,\rho^\dagger)$ yields $(\eta,\eta)$. At $(\eta,\eta)$, the
aggregate scores are
\[
s_i=(\tfrac{1+\varepsilon}{2},\tfrac{1-\varepsilon}{2}),
\qquad
s_j=(\tfrac{1-\varepsilon}{2},\tfrac{1+\varepsilon}{2}),
\]
so $\tau_i=\rho$ and $\tau_j=\rho^\dagger$. One bounded step therefore
returns to $(\rho,\rho^\dagger)$. The score gaps at the two states of
this Variant~S orbit are $1-2\varepsilon$ and $\varepsilon$,
respectively. Hence the aggregate scores remain at
$\ell^\infty$-distance at least
$\min\{\varepsilon,1-2\varepsilon\}/2$ from the tie hyperplane,
proving part~(b).
\end{proof}

The two-node forcing gadget of Theorem~\ref{thm:period2_forcing} is
illustrated in Figure~\ref{fig:forcing_gadget}.

\begin{figure}[h!]
\centering
\includegraphics{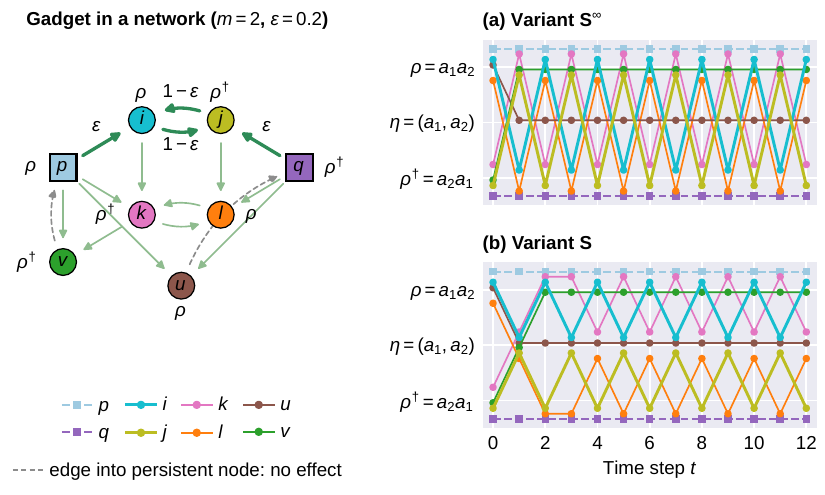}
\caption{The two-node forcing gadget of Theorem~\ref{thm:period2_forcing} embedded in a larger social network ($m=2$, $\varepsilon=0.2$). Left: persistent nodes $p$ and $q$ (squares) hold $\rho=a_1a_2$ and $\rho^{\dagger}=a_2a_1$. The gadget $i,j$ has $w_{ij}=w_{ji}=1-\varepsilon$, $w_{ip}=w_{jq}=\varepsilon$, and initial state $(\rho,\rho^{\dagger})$. The other free nodes have weights $w_{ki}=0.4$, $w_{kl}=0.3$, $w_{kp}=0.3$; $w_{lj}=0.4$, $w_{lk}=0.3$, $w_{lq}=0.3$; $w_{up}=w_{uq}=0.5$; $w_{vp}=0.6$, $w_{vk}=0.4$. Each free node's initial state is shown next to it, and an arrow $s\to t$ means that $s$ influences $t$. Dashed edges into $p$ and $q$ have no effect, because persistent nodes are pinned. Right: (a) under Variant~S$^{\infty}$ the gadget alternates between $(\rho,\rho^{\dagger})$ and $(\rho^{\dagger},\rho)$; (b) under Variant~S it alternates between $(\rho,\rho^{\dagger})$ and $(\eta,\eta)$, where $\eta=(a_1,a_2)$. In both variants $k$ and $l$ are drawn into the period-$2$ oscillation, $u$ settles at $\eta$, and $v$ settles at $\rho$. Curves are slightly offset vertically for visibility.}
\label{fig:forcing_gadget}
\end{figure}

\begin{rem}[Even periods and $2k$--oscillations]
\label{rem:2k}
Theorem~\ref{thm:period2_forcing} formalizes the intuition that contrarian camps interact with a period--2 (bipartite) influence kernel to produce \emph{even-period} behavior. More generally, if the two-step return map on one cyclic class admits a $k$-cycle, then by bipartite lifting the full synchronous dynamics admits a periodic orbit returning after $2k$ steps, of exact period $2k$ whenever $k$ is even (Theorem~\ref{thm:even_period_lift}). This mechanism explains why, in simulations, even periods (including $4$ and higher) are naturally associated with bipartite or near-bipartite structure.
\end{rem}
\section{Robustness (structural stability) away from tie hyperplanes}

\begin{prop}[Robustness of periodic orbits under small weight perturbations]
\label{prop:robustness}
Consider Variant~S with fixed $P$ and deterministic tie-breaking in $T$ and $\textsf{Step}$.
Let $\sigma(t)$ be a periodic orbit of period $p$.
Assume that along the orbit, every node’s aggregated score $s_i(t)$ remains at distance at least $\delta>0$ from all tie hyperplanes of $T$.
Then there exists $\varepsilon>0$ such that for any weight matrix $\widetilde W$ satisfying
\[
\lVert \widetilde W - W \rVert_{\infty} < \varepsilon .
\]
and still row-stochastic and compatible with the same directed edge set, the same periodic symbolic dynamics (i.e., the same targets and $\textsf{Step}$ choices) persists. In particular, a period-$p$ orbit exists for $\widetilde W$ as well. Here $\lVert\cdot\rVert_{\infty}$ denotes the entrywise maximum norm,
$\lVert A \rVert_{\infty}=\max_{i,j}|A_{ij}|$.
\end{prop}

\begin{proof}
On each region cut out by tie hyperplanes, $T$ is constant.
A uniform margin $\delta$ implies that sufficiently small perturbations of $W$ perturb each $s_i(t)$ by less than the distance to the nearest tie hyperplane, so targets are unchanged on the orbit.
With unchanged targets and deterministic $\textsf{Step}$, the orbit persists.
\end{proof}

\section{Variant A: asynchronous updates} 

Variant~A replaces the deterministic map $F_P$ by a (typically) irreducible Markov chain on the finite state space $\Omega^U$, with transition structure determined by the (possibly random or adversarial) update schedule \citep{istrate_marathe_ravi2012,macauley_mortveit2009}. Even on bipartite graphs, the two-step decoupling of Lemma~\ref{lem:bipartite_decoupling} need not persist under asynchronous (single-site) updates, so the parity mechanism of Theorem~\ref{thm:even_period_lift} is specific to synchronous dynamics \citep{nareddy_machta2020,macauley_mortveit2009}. Nevertheless, Variant~A can still exhibit long transients (metastable behavior) and may cycle on subsets via non-absorbing recurrent classes of the induced Markov chain \citep{cirillo_jacquier_spitoni2022,norris1997}.

A general sufficient condition for almost-sure convergence of Variant~A is the existence of a strict Lyapunov (potential) function that decreases under every single-node update; in that case, the induced asynchronous dynamics cannot visit any state infinitely often except fixed points, and thus reaches a fixed point in finite time almost surely \citep{norris1997,blume1993}. Identifying such Lyapunov functions for bounded Borda dynamics (beyond special symmetric cases) remains an open direction.

\section{Discussion and Conclusion}

This paper develops a mathematical framework for understanding preference dynamics on networks through bounded Borda aggregation. By coupling a social influence network with the discrete geometry of weak-order preference spaces, we obtain a tractable yet rich model that exhibits both consensus and oscillatory behavior.

\subsection{Main findings}

Our analysis reveals two distinct mechanisms for generating nontrivial periodic orbits. First, \emph{self-sustained oscillations} emerge endogenously from the interplay of directed influence cycles in $G$ and cycles in the preference move graph $H$, even without external persistent sources (Theorem~\ref{thm:self_traveling_wave}). 
These self-oscillations arise from the directed-cycle topology of the influence
network $G$ and the cycle structure of the preference graph $H$; as established
in Proposition~\ref{prop:unbounded_fixed_periodic}(c), they persist whether the
update step is bounded or unbounded.

Second, \emph{forced oscillations} arise when contrarian persistent camps pin antipodal preference orderings and interact through a bipartite network structure. The key structural feature is the bipartite (period-$2$) structure of the free-to-free influence kernel: nodes in partition $A$ receive aggregated signals from partition $B$ and vice versa. Its linear-algebraic signature is the eigenvalue $-\rho(W_{CC})$ (equal to $-1$ only when the free block receives no weight from persistent nodes). Combined with the antipodal camps, this parity structure generates period-$2$ oscillations among free nodes, as shown by the explicit construction of Theorem~\ref{thm:period2_forcing}. Bipartite lifting (Theorem~\ref{thm:even_period_lift}) turns a $k$-cycle of the two-step return map into an orbit of the full dynamics that returns after $2k$ steps, with exact period $2k$ whenever $k$ is even.

\subsection{Robustness and structural stability}

A notable feature of these oscillations is their robustness away from score-tie hyperplanes. Proposition~\ref{prop:robustness} establishes that periodic orbits persist under small perturbations of the weight matrix $W$, provided that aggregated scores maintain a uniform margin from tie boundaries. This structural stability reflects the fact that the dynamics are determined by discrete combinatorial choices (targets and bounded steps) once the scores clear the tie hyperplanes. Thus oscillatory behavior is not fragile but rather a stable feature of the underlying preference geometry and network topology.

\subsection{Contrast with asynchronous updating}

The synchronous (Variant S) and asynchronous (Variant A) dynamics exhibit fundamentally different behaviors. In Variant S, the two-step decoupling on bipartite graphs (Lemma~\ref{lem:bipartite_decoupling}) enables the parity mechanism that produces even-period orbits. In Variant A, where single nodes update sequentially, this decoupling is lost, and the dynamics typically converge to fixed points under mild conditions. This contrast highlights the importance of update timing and coordination in preference dynamics: synchronized collective updates can sustain oscillations that asynchronous individual updates cannot. 

\subsection{Broader implications and open directions}

From a social-science perspective, our results formalize conditions under which preference disagreement and polarization persist dynamically rather than resolving to consensus. The model captures several realistic features: (i) bounded rationality (agents move incrementally rather than jumping to optimal positions), (ii) ordinal outcomes (preferences are rankings, not scalars), (iii) relational structure (influence is mediated by a social network), and (iv) heterogeneous dispositions (persistent nodes represent stubborn actors or institutional anchors). The interplay of these elements can generate stable oscillations even without external forcing.

Natural extensions include: (a) multiple persistent camps with arbitrary cardinality and antipodal structure; (b) multi-camp partitions and higher-order network structures beyond bipartite graphs; (c) stochastic perturbations and noise to study robustness of oscillations; (d) adaptive network models where the influence graph itself evolves; (e) a deeper theoretical analysis of asynchronous updates (Variant A), the identification of parameter regimes admitting periodic or chaotic behavior, and the comparison of convergence rates and long-term dynamics between synchronous and asynchronous protocols; and (f) empirical validation on real preference-aggregation datasets from voting, committee decisions, or online deliberation platforms.

A key open problem is to characterize the full bifurcation structure of the bounded Borda dynamics as a function of network topology, weights, and the preference move graph. In particular, identifying Lyapunov functions or other global invariants for Variant A remains elusive and would enable stronger convergence guarantees. Finally, the relationship between oscillations in bounded Borda dynamics and periodic behavior in other preference-aggregation rules (e.g., majority dynamics, approval voting) deserves investigation.

Another important direction concerns the behavior of the dynamics
under structured preference domains. Fix an axis
\[
a_1\prec a_2\prec\cdots\prec a_m
\]
on $\Alts$. For weak orders, we use the single-plateaued extension of
single-peakedness: a nonempty interval of the axis may form a top
indifference plateau, while preferences decrease strictly away from it.

\begin{defn}[Single-plateaued weak orders]
A weak order $\sigma\in\Omega$ is \emph{single-plateaued} with respect
to the fixed axis if there exist $1\le \ell\le r\le m$ such that
\[
a_\ell\sim_\sigma a_{\ell+1}\sim_\sigma\cdots\sim_\sigma a_r,
\]
every alternative in $\{a_\ell,\dots,a_r\}$ is strictly preferred to
every alternative outside this set,
\[
a_q\succ_\sigma a_p\quad\text{whenever }1\le p<q\le\ell,
\]
and
\[
a_p\succ_\sigma a_q\quad\text{whenever }r\le p<q\le m.
\]
Write $\Omega_{\mathrm{SP}}\subseteq\Omega$ for the set of such weak
orders.
\end{defn}

We first ask whether Borda aggregation of single-plateaued weak-order
profiles yields single-plateaued targets. The following lemma
establishes this for $m=3$.

\begin{lem}[Borda aggregation preserves single-plateauedness for $m=3$]
\label{lem:sp_borda_m3}
Let $m=3$ with axis $a_1\prec a_2\prec a_3$ on $\Alts$.
If $\sigma_j\in\Omega_{\mathrm{SP}}$ for all $j\in V$ and $W$ is
row-stochastic, then
$\tau_i(\sigma)=T(s_i(\sigma))\in\Omega_{\mathrm{SP}}$
for every node $i\in V$.
\end{lem}

\begin{proof}
For $m=3$, the single-plateaued weak orders with respect to the stated
axis, together with their averaged Borda score vectors, are
\[
\begin{array}{c|c}
\sigma & b(\sigma) \\ \hline
a_1\succ a_2\succ a_3 & (2,1,0) \\
a_2\succ a_1\succ a_3 & (1,2,0) \\
a_2\succ a_3\succ a_1 & (0,2,1) \\
a_3\succ a_2\succ a_1 & (0,1,2) \\
a_1\sim a_2\succ a_3 & (\tfrac32,\tfrac32,0) \\
a_2\succ a_1\sim a_3 & (\tfrac12,2,\tfrac12) \\
a_2\sim a_3\succ a_1 & (0,\tfrac32,\tfrac32) \\
a_1\sim a_2\sim a_3 & (1,1,1).
\end{array}
\]
Each of these vectors satisfies
\[
2b(\sigma)_{a_2}\ge b(\sigma)_{a_1}+b(\sigma)_{a_3}.
\]
Because the weights are nonnegative and row-stochastic, the aggregate
score vector satisfies
\[
2(s_i)_{a_2}
=\sum_{j\in V}w_{ij}\,2b(\sigma_j)_{a_2}
\ge\sum_{j\in V}w_{ij}
\bigl(b(\sigma_j)_{a_1}+b(\sigma_j)_{a_3}\bigr)
=(s_i)_{a_1}+(s_i)_{a_3}.
\]
If $T(s_i)\notin\Omega_{\mathrm{SP}}$, then it is one of
\[
a_1\sim a_3\succ a_2,\qquad
a_1\succ a_3\succ a_2,\qquad
a_3\succ a_1\succ a_2,\qquad
a_1\succ a_2\sim a_3,\qquad
a_3\succ a_1\sim a_2.
\]
Every score vector having one of these five weak orders as its image
under $T$ satisfies
\[
2(s_i)_{a_2}<(s_i)_{a_1}+(s_i)_{a_3},
\]
a contradiction. Hence $T(s_i)\in\Omega_{\mathrm{SP}}$.
\end{proof}

If the state space is restricted to strict orders and elementary moves
are adjacent swaps, single-peakedness need not be preserved by a
bounded step. Indeed, for $m=3$ with axis
$a_1\prec a_2\prec a_3$, the path
\[
a_1\succ a_2\succ a_3
\;\longrightarrow\;
a_1\succ a_3\succ a_2
\;\longrightarrow\;
a_3\succ a_1\succ a_2
\;\longrightarrow\;
a_3\succ a_2\succ a_1
\]
is a geodesic in the strict-order adjacent-swap graph: it has length
$3$, equal to the number of pairwise reversals between its endpoints.
The endpoint orders are single-peaked, while the two intermediate orders
are not. This example concerns the strict-order adjacent-swap graph,
not the weak-order cover graph $H$ used in this paper. In $H$, the
corresponding route requires $6$ cover moves, whereas the two endpoint
orders have distance $4$, for example through the fully indifferent weak order. Thus this
example makes no claim about preservation of single-peakedness by
Variant~S in $H$.

For the unbounded map (Variant~S$^\infty$), the geodesic issue
disappears entirely: the agent jumps directly to $\tau_i(\sigma)$
without traversing any intermediate state.
Combined with Lemma~\ref{lem:sp_borda_m3}, this gives the following
invariance result.

\begin{prop}[Invariance of the single-plateaued domain under Variant~S$^\infty$, $m=3$]
\label{prop:sp_invariance}
Let $m=3$ with a fixed axis on $\Alts$.
Under Variant~S$^\infty$, if $\sigma(0)\in\Omega_{\mathrm{SP}}^V$,
then $\sigma(t)\in\Omega_{\mathrm{SP}}^V$ for all $t\ge 0$.
\end{prop}

\begin{proof}
By induction on $t$.
At each step, $\sigma_i(t+1)=\tau_i(\sigma(t))$ for nonpersistent
nodes $i\in U$.
If $\sigma(t)\in\Omega_{\mathrm{SP}}^V$, then by
Lemma~\ref{lem:sp_borda_m3}, $\tau_i(\sigma(t))\in\Omega_{\mathrm{SP}}$
for every $i\in U$.
Persistent nodes satisfy $\sigma_p(t)=\bar\sigma_p\in\Omega_{\mathrm{SP}}$
by assumption.
Hence $\sigma(t+1)\in\Omega_{\mathrm{SP}}^V$.
\end{proof}

\begin{rem}
An analogous invariance result for Variant~S (bounded step) would
require $\Omega_{\mathrm{SP}}$ to be geodesically convex in $H$, i.e.,
that every shortest path in $H$ between two single-plateaued weak orders
remain within $\Omega_{\mathrm{SP}}$. The strict-order adjacent-swap
example above does not resolve this question for the weak-order cover
graph $H$. We leave the bounded-update question on $H$ for future work.
\end{rem}

Restricting the dynamics to $\Omega_{\mathrm{SP}}$ effectively
prunes the available transitions in the move graph.
As a result, some of the mechanisms that generate oscillatory behavior
in the unrestricted domain---such as those based on directed cycles or
bipartite forcing configurations---may no longer be realizable.
Understanding how such structural restrictions affect convergence
properties remains an interesting direction for future work.

A related robustness question concerns the restriction to strict
preferences. In this case, the state space reduces to the set of
linear orders on $\Alts$, corresponding to the subset of $\Omega$
without indifferences. The bounded-step dynamics and the move
graph $H$ can be naturally restricted to this subset, yielding a
reduced transition structure in which ties are no longer permitted.

While this restriction eliminates certain transitions that rely on
indifferences, the qualitative features of the dynamics remain
similar. In particular, mechanisms leading to oscillatory behavior,
such as those induced by network structure, may still arise within
the strict domain, although their realization depends on the
reduced connectivity of the move graph.

The restriction to strict preferences can be viewed as a particular
instance of a broader class of \emph{$r$-step rules}: for any integer
$r\ge 1$, define $\textsf{Step}_r(\alpha,\beta)$ as the state reached
by moving at most $r$ steps along a shortest path in $H$ from $\alpha$
toward $\beta$ (with $\textsf{Step}_1 = \textsf{Step}$ and
$\textsf{Step}_r = \textsf{Jump}$ for $r\ge\mathrm{diam}(H)$).

\begin{prop}[$r$-step robustness of oscillations]
\label{prop:r_step_robustness}
Let $r\ge 1$ be any integer.
\begin{description}
  \item[(a)] Under the hypotheses of Theorem~\ref{thm:self_traveling_wave},
  the self-oscillation of period $k$ persists under $\textsf{Step}_r$ for
  every $r\ge 1$.
  \item[(b)] Under the hypotheses of Theorem~\ref{thm:period2_forcing},
  a period--$2$ orbit exists under $\textsf{Step}_r$ for every $r\ge1$.
  For $r=1$, it alternates between $(\rho,\rho^\dagger)$ and
  $(\eta,\eta)$; for $r\ge2$, it alternates between
  $(\rho,\rho^\dagger)$ and $(\rho^\dagger,\rho)$.
\end{description}
\end{prop}

\begin{proof}
\textbf{(a)} Under the hypotheses of Theorem~\ref{thm:self_traveling_wave},
the $H$-cycle adjacency condition requires that successive states
$\omega_r$ and $\omega_{r-1}$ on the $H$-cycle are adjacent in $H$,
i.e., $d_H(\omega_r,\omega_{r-1})=1$.
Hence $d_H(\sigma_i(t),\tau_i(\sigma(t)))=1\le r$ at every step of the
orbit, so $\textsf{Step}_r(\sigma_i(t),\tau_i(\sigma(t)))=\tau_i(\sigma(t))$
for any $r\ge 1$.
The orbit is therefore identical to the one produced by $\textsf{Jump}$,
and by Proposition~\ref{prop:unbounded_fixed_periodic}(c) it has period $k$.

\textbf{(b)} For $r=1$, the claim is part~(b) of
Theorem~\ref{thm:period2_forcing}. For $r\ge2$,
$d_H(\rho,\rho^\dagger)=2\le r$, so
$\textsf{Step}_r(\rho,\rho^\dagger)=\rho^\dagger$ and
$\textsf{Step}_r(\rho^\dagger,\rho)=\rho$. Hence the period--$2$
orbit is the Variant~S$^\infty$ orbit in part~(a) of
Theorem~\ref{thm:period2_forcing}.

\end{proof}

\begin{rem}
Proposition~\ref{prop:r_step_robustness} shows that the oscillatory results
are robust not only to the removal of the step-size bound (Proposition~\ref{prop:unbounded_fixed_periodic}) but to any relaxation of it.
The oscillations constructed in Theorems~\ref{thm:self_traveling_wave}
and~\ref{thm:period2_forcing} are features of the Borda-aggregation coupling
and the network topology, not of the specific value of the step-size parameter $r$.
A systematic comparison of convergence rates and transient behavior across
different values of $r$ is left for future work.
\end{rem}

\subsection{Conclusion}

This work demonstrates that the combination of network structure, ordinal state spaces, and bounded-step constraints creates a fertile ground for nontrivial dynamics. The Borda aggregation rule, grounded in classical social-choice theory, provides a principled and tractable aggregation mechanism that respects the geometry of preference orderings. By developing sufficient conditions for self-oscillations and forced oscillations in terms of directed topology, spectral properties, and the preference move graph, we bridge network dynamics and discrete preference geometry—a connection that may prove valuable for understanding polarization, consensus formation, and collective decision-making in complex social systems.

\bibliographystyle{unsrtnat}
\bibliography{references}

\end{document}